\documentclass[twocolumn,aps,pra,notitlepage,superscriptaddress,longbibliography]{revtex4-1}
\usepackage[breaklinks=true]{hyperref}
\usepackage{amsmath,amssymb,amsthm,mathtools,mathrsfs}
\usepackage[table,dvipsnames]{xcolor}%This is used to add colors to rows and columns in arrays (which I use for matrices) WARNING: Must load before tikz packages (i believe arrows.meta uses xcolor so there is a clash). or another option is to add "table" as a global option to the document class, ex. \documentclass[12pt,table]{article}
\usepackage{tikz}
\usepackage{adjustbox}
\usepackage{dsfont}%this is for mathbb{1} but you write mathds{1}
\usepackage[normalem]{ulem}
\usepackage{enumerate}
\usepackage[all]{xy} %This is to make higher categorical stuff

\theoremstyle{plain}
\newtheorem{theorem}{Theorem}
\newtheorem{claim}{Claim}
\newtheorem{lemma}{Lemma}

\theoremstyle{definition}
\newtheorem{definition}{Definition}

\newcommand{\bd}{\begin{definition}}
\newcommand{\ed}{\end{definition}}
\newcommand{\bt}{\begin{theorem}}
\newcommand{\et}{\end{theorem}}
\newcommand{\be}{\begin{equation}}
\newcommand{\ee}{\end{equation}}
\newcommand{\blem}{\begin{lemma}}
\newcommand{\elem}{\end{lemma}}

\DeclareMathAlphabet{\mathpzc}{OT1}{pzc}{m}{it} %this was too small so I used the following code from https://tex.stackexchange.com/questions/656585/pzc-has-a-smaller-size-than-mathcal
 \DeclareFontFamily{OT1}{pzc}{}
 \DeclareFontShape{OT1}{pzc}{m}{it}{ <-> s*[1.2] pzcmi7t }{}
 \DeclareMathAlphabet{\mathpzc}{OT1}{pzc}{m}{it}
 
\def\R{{{\mathbb R}}}

\newcommand{\Tr}{\operatorname{Tr}}

\def\E{\mathscr{E}}

\def\O{\mathscr{O}}

\newcommand{\norm}[1]{\left\lVert#1\right\rVert}

\begin{document}									
\preprint{APS/123-QED}

\title{Geometry from quantum temporal correlations}

\author{James Fullwood}
\email{fullwood@hainanu.edu.cn}
\affiliation{School of Mathematics and Statistics, Hainan University, 58 Renmin Avenue, Haikou, 570228, China}
\author{Vlatko Vedral}
\affiliation{Clarendon Laboratory, University of Oxford, Parks Road, Oxford OX1 3PU, United Kingdom}

\date{\today}

\begin{abstract}
In this work, we show how Euclidean 3-space uniquely emerges from the structure of quantum temporal correlations associated with sequential measurements of Pauli observables on a single qubit. Quite remarkably, the quantum temporal correlations which give rise to geometry are independent of the initial state of the qubit, which we show enables an observer to extract geometric data from sequential measurements without the observer having any knowledge of initial conditions. Such results suggest the plausibility that space itself may emerge from quantum temporal correlations, and we formulate a toy model of such a hypothetical phenomenon. 
\end{abstract}

\maketitle

\section{Introduction} 
There is a growing consensus in theoretical physics that spacetime is not a primitive notion \cite{Wheeler_1989}. In particular, opposed to a process of quantization where our quantum descriptions of nature are derived from a classical starting point which assumes the existence of space and time, it is now widely held that the notions of space and time should emerge from a more fundamental description of reality, whether it be a quantum description or something else altogether \cite{Seiberg_2006,Cao_2016,Arkani_2013}. 

While there are various proposals for realizing such a viewpoint, in this work we show how three-dimensional Euclidean geometry uniquely emerges from the structure of quantum \emph{temporal} correlations associated with sequential measurements of spin observables on a single qubit. Mathematically speaking, we show that given a matrix representation of the Clifford algebra of $\R^3$, then the Euclidean metric on $\R^3$ manifests itself in the representation as the two-time correlation function on the three-dimensional space spanned by the Pauli spin matrices. Quite remarkably, it turns out that the two-time correlation function which gives rise to the Euclidean metric on the space of Pauli observables is independent of the initial state of the qubit prior to measurement, which we show enables an observer to infer such geometric structure without any knowledge of initial conditions.  

As there has been much work inspired by the curious fact that the state space of a single qubit is the 3-dimensional Bloch ball \cite{Penrose_1967,vonW_2006,Wootters_1980,M_ller_2013}, our results reveal yet another intriguing connection between a quantum bit of information and the local spatial geometry of our world. Moreover, as this connection between quantum information and space relies crucially on temporal correlations, our results link together at once the notions of quantum information, space, and time. 

\section{A mathematical theorem} 
We now formulate a theorem which serves as the mathematical foundation for our results, whose proof will be given in a later section. To set the stage, let $A$ denote a finite-dimensional quantum system, and let $\bold{Obs}(A)$ denote the algebra of observables on $A$, whose underlying set consists of all Hermitian operators on the associated Hilbert space $\mathcal{H}_A$. A real linear subspace $\mathfrak{S}\subset \bold{Obs}(A)$ will be referred to as a \emph{$\gamma$-space} if there exists a basis $\mathcal{B}$ of $\mathfrak{S}$ such that $\O_i\O_j+\O_j\O_i=2\delta_{ij}\mathds{1}$ for all $\O_i,\O_j\in \mathcal{B}$, where $\mathds{1}\in \bold{Obs}(A)$ denotes the identity operator. Such a basis $\mathcal{B}$ will be referred to as a \emph{$\gamma$-basis} for $\mathfrak{S}$. We note the the algebra generated by a $\gamma$-space of dimension $d>0$ yields an operator representation of the real Clifford algebra generated by $\R^d$ with respect to the Euclidean metric. In such a case, the associated quantum system $A$ may be viewed as a system consisting of ${\lfloor d/2 \rfloor}$ qubits. The prototypical example of a $\gamma$-space is the three-dimensional space $\mathfrak{Pauli}$ spanned by the Pauli spin matrices $\sigma_x$, $\sigma_y$ and $\sigma_z$, which are a $\gamma$-basis of $\mathfrak{Pauli}$.

Now let $\rho$ be a density operator on $\mathcal{H}_A$, and let $\mathfrak{S}\subset \bold{Obs}(A)$ be a linear subspace, possibly different from a $\gamma$-space. Suppose that the system $A$ is initially prepared in state $\rho$, after which sequential measurements are performed at times $t_1$ and $t_2>t_1$ of observables $\O_1$ and $\O_2$ in $\mathfrak{S}$. If the system $A$ evolves trivially between measurements, then the \emph{two-time correlation function} with respect to the initial state $\rho$ is the function $\mathscr{E}_{\rho}:\mathfrak{S}\times \mathfrak{S}\to \R$ given by \cite{Lostaglio_2023,FuPa24}
\begin{equation} \label{2TXCFX87}
\mathscr{E}_{\rho}(\O_1,\O_2)=\sum_{i}\lambda_i \Tr\Big[P_i \rho P_i \O_2\Big]\, ,
\end{equation}
where $\O_1=\sum_{i}\lambda_i P_i$ is the spectral decomposition of $\O_1$ and $\{\lambda_i\}$ is the set of \emph{distinct} eigenvalues of $\O_1$. Our main result is the following:
\begin{theorem} \label{TMX1}
A real linear subspace $\mathfrak{S}\subset \bold{Obs}(A)$ is a $\gamma$-space if and only if for every state $\rho$ on $A$, the two-time correlation function  $\mathscr{E}_{\rho}:\mathfrak{S}\times \mathfrak{S}\to \R$ given by \eqref{2TXCFX87} yields a real inner product on $\mathfrak{S}$, independent of $\rho$. 
\end{theorem}

We note that while there are known results about quantum temporal correlations which are equivalent to the $\implies$ statement of Theorem~\ref{TMX1} in the case of a single qubit \cite{BTCV04,Fritz10,Emary_2013}, no connection to geometry seems to have ever been articulated. In any case, the fact that the two-time correlation function $\mathscr{E}_{\rho}$ yields an inner product on $\gamma$-spaces is quite remarkable, since for general spaces of observables $\mathscr{E}_{\rho}$ is not linear in its first argument, or even symmetric \cite{FuPa24}. Moreover, the fact that $\mathscr{E}_{\rho}$ is independent of $\rho$ for $\gamma$-spaces is perhaps even more remarkable, as it renders the knowledge of an initial state as being superfluous information for any process consisting of sequential measurements of $\gamma$-observables. 

Since a $\gamma$-space of dimension $d=3$ is necessarily isomorphic to the space $\mathfrak{Pauli}$ spanned by the Pauli spin matrices, it follows from Theorem~\ref{TMX1} that the space $\mathfrak{Pauli}$ is the unique space of observables on a single qubit whose temporal correlations yield a Euclidean structure on its underlying vector space. Conversely, it follows from Theorem~\ref{TMX1} that Euclidean 3-space is the unique geometry which emerges from temporal correlations between sequential measurements of a single qubit. As such, Theorem~\ref{TMX1} establishes a striking connection between the local geometry of our world and quantum temporal correlations whose significance is not yet well-understood. Nevertheless, this connection is succinctly captured by the equation
\[
\E(\vec{r}\cdot \sigma\,, \vec{s}\cdot \sigma)=\vec{r}\cdot \vec{s}\, ,
\]
where $\vec{r}\cdot \vec{s}$ denotes the dot product in $\R^3$, and we have suppressed the subscript on the two-time correlation function $\E$ due to its independence of the initial state $\rho$.

\section{A hypothetical model for the emergence of space} 
The usual meaning given to the Pauli spin matrices is that they represent observables associated with orthogonal directions in a pre-existing space, which for example may be identified with orientations of a Stern-Gerlach apparatus located in some laboratory. In light of Theorem~\ref{TMX1} however, it seems reasonable to suggest that perhaps the notion of \emph{space itself} may emerge from temporal correlations associated with sequential measurements of some yet to be discovered observables represented by the Pauli spin matrices. While it is not at all evident at this point how this may actually occur in nature, we nevertheless put forth a germ of an idea which we hope may eventually shed some light on the hypothetical phenomenon.

For this, imagine a universe consisting a single observer and a single qubit, and that there is no pre-existing space in this universe. We do assume however that the observer experiences a personal notion of time, and that the observer is equipped with some probe which enables them to measure the qubit in some abstract sense. The outcomes of such measurements will then be viewed as a form of ``sensory data" internal to the observer. We further assume that the possible measurements induced by this probe yield a set of observables mathematically represented by the $\gamma$-space $\mathfrak{Pauli}$, which we recall is the real linear span of the Pauli spin matrices $\sigma_x$, $\sigma_y$ and $\sigma_z$. We now utilize Theorem~\ref{TMX1} to show how the qubit may serve as a reservoir of quantum information from which Euclidean geometry emerges as the structure of temporal correlations associated with the sensory data resulting from the observer's measurements, thus forming the local geometry of their world. In particular, we show how the qubit contains the necessary information required for the observer to extract the Euclidean metric on $\mathfrak{Pauli}$ from a processing of quantum information associated with sequential measurements of the qubit, without the observer having any knowledge of the state of the qubit at any moment in time. 

So let $\vec{r}\cdot \sigma\, ,\,\vec{s}\cdot \sigma\in \mathfrak{Pauli}$ for some $\vec{r},\vec{s}\in \R^3$, and suppose the observer chooses to perform $N$ sequential measurements on the qubit at times $t_1<\cdots<t_N$, with $\vec{r}\cdot \sigma$ being measured at times $t=t_{2k-1}$ and $\vec{s}\cdot \sigma$ being measured at times $t=t_{2k}$ for all $k\in \{1,\ldots,\lceil N/2 \rceil \}$. We let $x_i$ denote the measurement outcome at time $t=t_i$ for all $i\in \{1,\ldots,N\}$. By Theorem~\ref{TMX1}, the two-time correlation function $\E_{\rho}$ as given by \eqref{2TXCFX87} is independent of $\rho$, from which it follows that the correlations between $x_i$ and $x_{i+1}$ are independent of the state of the qubit prior to measurement at time $t=t_i$, and are thus independent of $i$. Moreover, we also know from Theorem~\ref{TMX1} that $\E_{\rho}$ is an inner product, and is thus symmetric in its arguments. It then follows that $x_ix_{i+1}$ may be viewed as the product of a sequential measurement of $\vec{r}\cdot \sigma$ followed by $\vec{s}\cdot \sigma$, even if at time $t=t_i$ it was $\vec{s}\cdot \sigma$ that was measured instead of $\vec{r}\cdot \sigma$. As such, the products $x_ix_{i+1}$ may be viewed as the outcome of $N-1$ runs of a sequential measurement of $\vec{r}\cdot \sigma$ followed by $\vec{s}\cdot \sigma$, thus for every initial state $\rho$ we have
\[
\E_{\rho}(\vec{r}\cdot \sigma\,, \vec{s}\cdot \sigma)=\vec{r}\cdot \vec{s}\approx \frac{\sum_{i=1}^{N-1}x_ix_{i+1}}{N-1}  \qquad N\gg 1\, .
\]
Moreover, since $|x_1|=\norm{\vec{r}\,}$ and $|x_2|=\norm{\vec{s}\,}$, it follows that
\[
\theta\approx \cos^{-1}\left(\frac{\sum_{i=1}^{N-1}x_ix_{i+1}}{(N-1)\cdot |x_1x_2|}\right)\, ,
\]
where $\theta$ is the angle between $\vec{r}$ and $\vec{s}$. 

Therefore, as the observer interacts with the qubit by making sequential measurements, they may process the information in the form of successive products $x_ix_{i+1}$, which accumulates over time to form the information required for the structural realization of Euclidean 3-space. We note that if the measurements induced by the observer's probe are not accurately modeled by the space $\mathfrak{Pauli}$, then it follows from Theorem~\ref{TMX1} that the temporal correlations associated with the observer's measurements will \emph{not} yield a Euclidean structure on the space of observables induced by the probe. 

It is also important to emphasize here that while Theorem~\ref{TMX1} is a statement about \emph{two}-time correlations, we were able to exploit the fact that the two-time correlation function $\mathscr{E}_{\rho}$ on $\mathfrak{Pauli}$ is independent of $\rho$ by extracting $N-1$ two-time measurements from $N$ \emph{sequential} measurements, without ever having to re-prepare the state of the qubit. This fact is crucial, since the observer in our model can only ever probe the qubit, and thus does not have the ability to prepare the qubit in a fixed state. 

\section{Proof of Theorem~\ref{TMX1}} 
Let $A$ be a finite-dimensional quantum system, let $\mathfrak{S}\subset \textbf{Obs}(A)$ be a real linear subspace, and let $d$ denote the dimension of $\mathfrak{S}$. 
\\

\noindent \underline{($\implies$)}: Suppose $\mathfrak{S}$ is a $\gamma$-space, let $\{\mathscr{O}_1,\ldots,\mathscr{O}_d\}$ be a $\gamma$-basis for $\mathfrak{S}$, and let $\rho$ be a density operator on $\mathcal{H}_A$. We will now show that the two-time correlation function $\mathscr{E}_{\rho}:\mathfrak{S}\times \mathfrak{S}\to \R$ yields a real inner product on $\mathfrak{S}$ which is independent of $\rho$.

So let $\mathscr{O}=\sum_{i}\lambda_1\mathscr{O}_i$ be an arbitrary element of $\mathfrak{S}$.  We then have
\begin{align*}
\O^2&=\Big(\sum_{i}\lambda_i \O_i\Big)\Big(\sum_{j}\lambda_j \O_j\Big)=\sum_{i,j}\lambda_i\lambda_j \O_i\O_j \\
&=\sum_{i\neq j}\lambda_i\lambda_j (\O_i\O_j+\O_j\O_i)+\sum_{i}\lambda_i^2 \O_i^2 \\
&=\sum_{i\neq j}\lambda_i\lambda_j 2\delta_{ij}\mathds{1}+\sum_{i}\lambda_i^2 \mathds{1} \\
&=\sum_{i}\lambda_i^2 \mathds{1}\, ,
\end{align*}
so that the distinct eigenvalues of $\O$ are $\pm \lambda$, where $\lambda=\sqrt{\sum_{i}\lambda_i^2}$. Now let $\Pi^{\pm}$  be projection operators such that $\O=\lambda\Pi^+-\lambda \Pi^-$ and $\Pi^++\Pi^-=\mathds{1}$. It then follows that 
\begin{equation} \label{PJTXS67}
\Pi^{\pm}=\frac{1}{2}\Big(\mathds{1}\pm \frac{1}{\lambda}\O\Big)\, ,
\end{equation} 
thus for all $\mathscr{P}\in \bold{Obs}(A)$ we have
\begin{align*}
\mathscr{E}_{\rho}(\O,\mathscr{P})&=\lambda \Tr\Big[\Pi^+\rho\hspace{0.5mm} \Pi^+\mathscr{P}\Big]-\lambda \Tr\Big[\Pi^-\rho\hspace{0.5mm} \Pi^-\mathscr{P}\Big]  \\
&=\frac{\lambda}{4} \Tr\left[\frac{2}{\lambda}\Big\{\rho\, , \O\Big\}\mathscr{P}\right] &&  \\
&=\frac{1}{2}\Tr\left[\rho \hspace{0.25mm} \Big\{\O\, , \mathscr{P}\Big\}\right]\, ,
\end{align*}
where $\{\cdot\,,\cdot\}$ denotes the anti-commutator. Now let $\mathscr{P}=\sum_{j}\mu_j\mathscr{O}_j$ be the expansion with respect to the $\gamma$-basis of an element of $\mathfrak{S}$ possibly different from $\mathscr{O}$. We then have
\begin{align*}
\mathscr{E}_{\rho}(\O,\mathscr{P})&=\mathscr{E}_{\rho}\Big(\sum_{i}\lambda_i \O_i\, , \sum_{j}\mu_j \O_j\Big) \\
&=\frac{1}{2}\Tr\left[\rho \hspace{0.25mm} \Big\{\sum_{i}\lambda_i \O_i\, , \sum_{j}\mu_j \O_j\Big\}\right] \\
&=\frac{1}{2}\sum_{i,j}\lambda_i\mu_j \Tr\left[\rho \hspace{0.25mm} \Big\{ \O_i\, ,  \O_j\Big\}\right]  \\
&=\frac{1}{2}\sum_{i,j}\lambda_i\mu_j\Tr\left[\rho \hspace{0.25mm} 2\delta_{ij}\mathds{1}\right]  \\
&=\sum_{i}\lambda_i\mu_i\, , 
\end{align*}
thus $\mathscr{E}_{\rho}$ is a real inner product on $\mathfrak{S}$ which is independent of $\rho$, as desired.

We now prove the converse. An observable will be referred to as \emph{dichotomic} if and only if its set of distinct eigenvalues is $\{\pm 1\}$. 

\noindent \underline{($\impliedby$)}: Suppose that for every density operator $\rho$ on $\mathcal{H}_A$, the two-time correlation function $\mathscr{E}_{\rho}:\mathfrak{S}\times \mathfrak{S}\to \R$ given by \eqref{2TXCFX87} yields a real inner product on $\mathfrak{S}$ which is independent of $\rho$, and let $\{\O_1,\ldots,\O_d\}$ be an orthonormal basis of $\mathfrak{S}$ with respect to the real inner product induced by $\mathscr{E}_{\rho}$. We will now show that these assumptions imply that $\mathfrak{S}$ is a $\gamma$-space, and that $\{\O_1,\ldots,\O_d\}$ is a $\gamma$-basis for $\mathfrak{S}$.

\begin{claim} \label{C1}
$\O_i$ is dichotomic for all $i\in \{1,\ldots,d\}$.
\end{claim}
\begin{proof}
Let $\mathscr{O}\in \mathbf{Obs}(A)$ be an observable, and suppose $\mathscr{O}=\sum_{i}\lambda_i P_i$ is the spectral decomposition of $\mathscr{O}$. We then have
\begin{align*}
\mathscr{E}_{\rho}(\mathscr{O},\mathscr{O})&=\sum_{i}\lambda_i \Tr\Big[P_i \rho P_i \mathscr{O} \Big]=\sum_{i}\lambda_i \Tr\Big[\rho P_i \mathscr{O}P_i \Big]  \\
&=\sum_{i}\lambda_i \Tr\Big[\rho\hspace{0.5mm} \lambda_i P_i \Big]=\Tr\Big[\rho \sum_{i}\lambda_i^2 P_i\Big] \\
&=\Tr\Big[\rho \hspace{0.25mm} \mathscr{O}^2\Big] \, ,
\end{align*}
so that for all $\mathscr{O}\in \mathfrak{S}$ we have
\begin{equation} \label{NMEQX37}
\E_{\rho}(\mathscr{O},\mathscr{O})=\Tr\Big[\rho \hspace{0.25mm} \mathscr{O}^2\Big] \, .
\end{equation}
From the assumption that $\{\mathscr{O}_i\}_{i=1}^{d}$ is an orthonormal basis of $\mathfrak{S}$ with respect to the inner product induced by $\mathscr{E}_{\rho}$, together with the assumption that $\E_{\rho}$ is independent of $\rho$, it then follows that for every state $\rho$ on $A$ we have 
$\Tr[\rho \hspace{0.25mm} \mathscr{O}_i^2]=1$. As the identity operator $\mathds{1}$ is the only operator $X$ such that $\Tr[\rho X]=1$ for all states $\rho$ on $A$, it follows that $\mathscr{O}_i^2=\mathds{1}$ for all $i\in \{1,\ldots,d\}$, thus $\mathscr{O}_i=\pm \mathds{1}$ or $\mathscr{O}_i$ is dichotomic for all $i\in \{1,\ldots,d\}$.

We now claim that $\mathscr{O}_i\neq \pm \mathds{1}$ for all $i\in \{1,\ldots,d\}$. Indeed, for a contradiction, suppose without loss of generality that $\mathscr{O}_1=\pm \mathds{1}$. It then follows that the set of distinct eigenvalues of $\mathscr{O}_2$ is necessarily $\{\pm 1\}$, so that $\O_2=\Pi^{+}-\Pi^{-}$ for some orthogonal projection operators $\Pi^{\pm}$ such that $\Pi^{+}+\Pi^{-}=\mathds{1}$. Then for a given initial state $\rho$, we have 
\begin{align*}
\E_{\rho}(\O_1,\O_2)&=\E_{\rho}(\pm \mathds{1},\O_2)=\pm \E_{\rho}(\mathds{1},\O_2) \\
&=\pm \Tr\Big[\mathds{1}\rho \mathds{1}\O_2\Big]=\pm \Tr\Big[\rho \hspace{0.25mm} \O_2\Big]\, .
\end{align*}
Now since $\O_1$ and $\O_2$ are assumed to be orthogonal with respect to the inner product induced by $\E_{\rho}$, and since $\E_{\rho}$ is independent of $\rho$, it follows that $\Tr[\rho \hspace{0.25mm} \O_2]=0$ for every state $\rho$ on $A$. Writing $\Pi^{+}=\sum_i P_i$ as a sum of orthogonal 1-dimensional projectors and taking $\rho=P_{j}$ for some $j$, we then have
\begin{align*}
0&=\Tr\Big[P_j \hspace{0.25mm} \O_2\Big]=\Tr\Big[P_j(\Pi^{+}-\Pi^{-})\Big] \\
&=\Tr\Big[P_j\Big(\sum_i P_i-\Pi^{-}\Big)\Big]=\Tr\Big[P_{j}^2\Big]=\Tr\Big[P_j\Big]=1\, ,
\end{align*}
which is obviously absurd as $0\neq 1$. It then follows that $\mathscr{O}_i\neq \pm \mathds{1}$ for all $i\in \{1,\ldots, d\}$, thus $\mathscr{O}_i$ is in fact dichotomic for all $i\in \{1,\ldots, d\}$, as desired.
\end{proof}

Since $\O_i$ is dichotomic for all $i$, we let $\Pi_{i}^{\pm}$ denote the associated projection operators onto the $\pm 1$ eigenspaces of $\O_i$, so that for all $i\in \{1,\ldots,d\}$ we have $\O_i=\Pi_{i}^{+}-\Pi_{i}^{-}$ and $\Pi_{i}^{+}+\Pi_{i}^{-}=\mathds{1}$.

\begin{claim}\label{C3}
Let $i,j\in \{1,\dots,d\}$ with $i\neq j$. Then $\frac{\O_i+\O_j}{\sqrt{2}}$ is dichotomic.
\end{claim}

\begin{proof}
Let $i,j\in \{1,\dots,d\}$ with $i\neq j$. For an arbitrary initial state $\rho$ we have
\[
\Tr\Big[\rho \hspace{0.25mm} (\O_i+\O_j)^2\Big]=\E_{\rho}(\O_i+\O_j,\O_i+\O_j)=2\, ,
\]
where the first equality follows from equation \eqref{NMEQX37} and the second equality follows from the Pythagorean Theorem. We then have
\[
\Tr\left[\rho \hspace{0.25mm} \left(\frac{\O_i+\O_j}{\sqrt{2}}\right)^2\right]=1
\]
for all states $\rho$ on $A$, from which it follows that either
\begin{equation} \label{DMX747}
\frac{\O_i+\O_j}{\sqrt{2}}=\pm \mathds{1}
\end{equation}
or $(\O_i+\O_j)/\sqrt{2}$ is dichotomic. However, since $\O_i$ and $\O_j$ are dichotomic their traces are both integers, so that taking the trace of both sides of equation \eqref{DMX747} yields a contradiction, thus $(\O_i+\O_j)/\sqrt{2}$ is necessarily dichotomic, as desired.
\end{proof}

Since $(\O_i+\O_j)/\sqrt{2}$ is dichotomic for all $i$ and $j$ with $i\neq j$, we let $\Pi_{ij}^{\pm}$ denote the associated projection operators onto the $\pm 1$ eigenspaces of $(\O_i+\O_j)/\sqrt{2}$, so that  
\[
\frac{\O_i+\O_j}{\sqrt{2}}=\Pi_{ij}^{+}-\Pi_{ij}^{-} \quad \text{and} \quad \Pi_{ij}^{+}+\Pi_{ij}^{-}=\mathds{1}
\]
for all $i,j\in \{1,\ldots,d\}$ with $i\neq j$. Now let $i,j\in \{1,\dots,d\}$ with $i\neq j$, so that we have
\begin{equation} \label{OIXJS47}
\O_i+\O_j=\sqrt{2}(\Pi_{ij}^{+}-\Pi_{ij}^{-})\, .
\end{equation}
Multiplying both sides of equation \eqref{OIXJS47} on the left by $\O_i$ and using the fact that $\O_i^2=\mathds{1}$ yields
\begin{equation} \label{IJX37}
\O_i\O_j=\sqrt{2}\O_i(\Pi_{ij}^{+}-\Pi_{ij}^{-})-\mathds{1}\, ,
\end{equation}
and similarly we have 
\begin{equation} \label{JXI37}
\O_j\O_i=\sqrt{2}\O_j(\Pi_{ij}^{+}-\Pi_{ij}^{-})-\mathds{1}\, .
\end{equation}
Adding equations \eqref{IJX37} and \eqref{JXI37} then yields
\begin{equation} \label{ANTXCM81}
\{\O_i\, ,\O_j\}=(\O_i+\O_j)^2-2\mathds{1}\, .
\end{equation}
Now since
\begin{align*}
(\O_i+\O_j)^2&=2(\Pi_{ij}^{+}-\Pi_{ij}^{-})^2 \\
&=2\Big((\Pi_{ij}^{+})^2-\Pi_{ij}^{+}\Pi_{ij}^{-}-\Pi_{ij}^{-}\Pi_{ij}^{+}+(\Pi_{ij}^{-})^2\Big) \\
&=2(\Pi_{ij}^{+}+\Pi_{ij}^{-}) \\
&=2\mathds{1}\, ,
\end{align*}
it follows from \eqref{ANTXCM81} that $\{\O_i,\O_j\}=0$ for $i\neq j$. This together with the fact that $\O_i^2=\mathds{1}$ for all $i$ yields the anti-commutator $\{\O_i\, ,\O_j\}=2\delta_{ij}\mathds{1}$ for all $i,j\in \{1,\ldots,d\}$, thus $\mathfrak{S}$ is a $\gamma$-space and $\{\mathscr{O}_1,\ldots,\mathscr{O}_d\}$ is a $\gamma$-basis for $\mathfrak{S}$, as desired.

\section{Concluding remarks} 
In this work, we established an intriguing connection between a quantum bit of information, space, and time. In particular, we showed that the the structure of temporal correlations between sequential measurements of Pauli observables on a single qubit is equivalent to the Euclidean metric on the three-dimensional space $\mathfrak{Pauli}$ spanned by the Pauli spin matrices. Conversely, we showed that the space $\mathfrak{Pauli}$ is in fact uniquely characterized by this property, as we proved that $\mathfrak{Pauli}$ is the \emph{only} space of observables on a single qubit whose two-time correlation function is precisely the Euclidean metric. 

A conspicuous feature of the temporal correlations which give rise to geometry is the fact that they are independent of initial conditions, which enables an observer to extract the associated geometric structure from sequential measurements of a qubit without ever having to prepare the qubit in a fixed state. With this in mind, we formulated a toy model of how Euclidean 3-space may emerge from quantum temporal correlations in a hypothetical universe consisting of a single qubit and a single observer. In our model, the qubit exists in a realm without any pre-existing space, as it serves as a reservoir from which an observer may extract the information necessary for a geometric structure to emerge from the correlations between their sequential measurements of the qubit over time. Of course the question remains as to how the metric of space\emph{time} may emerge in such a framework, and more generally, how gravity may enter the picture. As the viewpoint esposed in this work sees geometry as a structure on a space of observables, perhaps the recent work \cite{Bressanini_24} on \emph{quantum observables over time} may provide an avenue for future investigation along these lines.

Upon discovering the fundamental law of the lever, Archimedes is said to have summarized his discovery via the poetic statement ``Give me a lever long enough and a fulcrum on which to place it, and I shall move the world.". In a similar vein, we summarize the results of this Letter with the following statement: 

``Give me a qubit for long enough and a probe in which to measure it, and I shall extract the geometry of our world."

\textbf{Acknowledgements} JF is supported by a Hainan University startup fund for the project ``Spacetime from Quantum Information", and would like to thank Giulio Chiribella and Edgar Guzm\'an Gonz\'alez for useful discussions. VV acknowledges support from the Templeton and the Gordon and Betty Moore foundations.

%\clearpage
%\newpage

\bibliography{references}

\end{document}